\documentclass[11pt]{article}

\usepackage{fullpage}
\usepackage{amssymb,amsmath}
\usepackage{amsthm}
\usepackage{graphicx}
\usepackage{subfig}
\usepackage{algorithm}
\usepackage{algorithmic}
\usepackage{array}
\usepackage{xspace}
\usepackage{color,soul}
\usepackage{paralist,mdwlist}

\usepackage{tikz}
\usetikzlibrary{arrows,decorations.pathmorphing,backgrounds,positioning,fit,through}

\definecolor{nym-blue}{HTML}{003581}
\definecolor{nym-orange}{HTML}{F47937}
\definecolor{whitesmoke}{HTML}{F5F5F5}

\newtheorem{lemma}{Lemma}
\newtheorem{observation}[lemma]{Observation}
\newtheorem{theorem}{Theorem}
\newtheorem{definition}{Definition}

\newcommand{\eqdf}{\stackrel{\scriptscriptstyle \triangle}{=}}

\newcommand{\set}[1]{\left\{ #1 \right\}}
\newcommand{\paren}[1]{\left( #1 \right)}

\newcommand{\inv}[1]{\frac{1}{#1}}

\newcommand{\ceil}[1]{\left\lceil {#1} \right\rceil}
\newcommand{\floor}[1]{\left\lfloor {#1} \right\rfloor}
\newcommand{\half}{\frac{1}{2}}
\newcommand{\threehalves}{\frac{3}{2}}

\newcommand{\naturals}{\mathbb{N}}
\newcommand{\rationals}{\mathbb{Q}}
\newcommand{\reals}{\mathbb{R}}

\newcommand{\strip}{\textsc{Strip Cover}\xspace}
\newcommand{\sosc}{\textsc{OnceSC}\xspace}
\newcommand{\srsc}{\textsc{RadSC}\xspace}
\newcommand{\stsc}{\textsc{TimeSC}\xspace}
\newcommand{\sosclong}{\textsc{Set Once Strip Cover}\xspace}
\newcommand{\srsclong}{\textsc{Set Radius Strip Cover}\xspace}
\newcommand{\stsclong}{\textsc{Set Time Strip Cover}\xspace}

\newcommand{\partition}{\textsc{Partition}\xspace}

\newcommand{\rr}{\textsc{RoundRobin}\xspace}
\newcommand{\RR}{\textsc{RR}\xspace}
\newcommand{\all}{\textsc{All}}
\newcommand{\opt}{\textsc{Opt}\xspace}

\begin{document}

\title{\bf Set It and Forget It: \\
Approximating the Set Once Strip Cover Problem}

\author{%
Amotz Bar-Noy \\
Department of Computer Science \\
The Graduate Center of the CUNY \\
New York, NY 10016, USA \\
\texttt{amotz@sci.brooklyn.cuny.edu}
\and
Ben Baumer \\
Department of Mathematics \& Statistics \\
Smith College \\
Northampton, MA 01063, USA \\
\texttt{bbaumer@smith.edu}
\and
Dror Rawitz \\
School of Electrical Engineering \\
Tel Aviv University \\
Tel-Aviv 69978, Israel \\
\texttt{rawitz@eng.tau.ac.il}
}


\begin{titlepage}

\maketitle

\begin{abstract}
We consider the \sosclong problem, in which $n$ wireless sensors are
deployed over a one-dimensional region. Each sensor has a fixed
battery that drains in inverse proportion to a radius that can be set
just once, but activated at any time.  The problem is to find an
assignment of radii and activation times that maximizes the length of
time during which the entire region is covered.  We show that this
problem is NP-hard. Second, we show that \rr, the algorithm in which
the sensors take turns covering the entire region, has a tight
approximation guarantee of $\threehalves$ in both \sosclong and the
more general \strip problem, in which each radius may be set
finitely-many times. Moreover, we show that the more general class of
\emph{duty cycle} algorithms, in which groups of sensors take turns
covering the entire region, can do no better. Finally, we give an
optimal $O(n^2 \log{n})$-time algorithm for the related \srsclong
problem, in which sensors must be activated immediately.
\end{abstract}


\medskip
\noindent
\textbf{Keywords}:
wireless sensor networks,
strip cover,
barrier coverage,
network lifetime.

\thispagestyle{empty}
\end{titlepage}


\section{Introduction}

Suppose that $n$ sensors are deployed over a one-dimensional
region that they are to cover with a wireless network. Each sensor is
equipped with a finite battery charge that drains in inverse
proportion to the sensing radius that is assigned to it, and each
sensor can be activated only once.  In the \sosclong (\sosc) problem,
the goal is to find an assignment of radii and activation times that
maximizes the \emph{lifetime} of the network, namely the length of
time during which the entire region is covered.

Formally, we are given as input the locations $x \in [0,1]^n$ and
battery charges $b \in \mathbb{Q}^n$ for each of $n$ sensors. While we
cannot move the sensors, we do have the ability to set the sensing
radius $\rho_i$ of each sensor and the time $\tau_i$ when it should
become active.  Since each sensor's battery drains in inverse
proportion to the radius we set (but cannot subsequently change), each
sensor covers the region $[x_i - \rho_i, x_i + \rho_i]$ for
$b_i/\rho_i$ time units.  Our task is to devise an algorithm that
finds a schedule $S = (\rho, \tau) \in [0,1]^n \times [0,\infty)^n$
for any input $(x,b)$, such that $[0,1]$ is completely covered for as
long as possible.
%
%


\paragraph*{\bf Motivation.}
Scheduling problems of this ilk arise in many applications,
particularly when the goal is \emph{barrier coverage}
(see~\cite{cardei2004coverage,wang2006survey} for surveys,
or~\cite{kumar2007barrier} for motivation).  Suppose that we have a
highway, supply line, or fence in territory that is either hostile or
difficult to navigate.  While we want to monitor activity along this
line, conditions on the ground make it impossible to systematically
place wireless sensors at specific locations.  However, it is feasible
and inexpensive to deploy adjustable range sensors along this line by,
say, dropping them from an airplane flying overhead
(e.g.~\cite{cardei2005improving,saipulla2009barrier,taniguchi2011uniform}).
Once deployed, the sensors send us their location via GPS, and we wish
to send a single radius-time pair to each sensor as an assignment.
Replacing the battery in any sensor is infeasible. How do we construct
an assignment that will keep this vital supply line completely
monitored for as long as possible?


\paragraph*{\bf Models.}
While the focus of this paper is the \sosc problem, we touch upon
three closely related problems. In each problem the location and
battery of each sensor are fixed, and a solution can be viewed as a
finite set of radius-time pairs.
In \sosc, both the radii and the activation times are variable, but
can be set only once. In the more general \strip problem, the radius
and activation time of each sensor can be set finitely many times. On
the other hand, if the radius of each sensor is fixed and given as
part of the input, then we call the problem of assigning an activation
time to each sensor so as to maximize network lifetime \stsclong
(\stsc).  \srsclong (\srsc) is another variant of \sosc in which all
of the sensors are scheduled to activate immediately, and the problem
is to find the optimal radial assignment.
Figure~\ref{fig:schematic} summarizes the important
differences between related problems and illustrates their
relationship to one another.

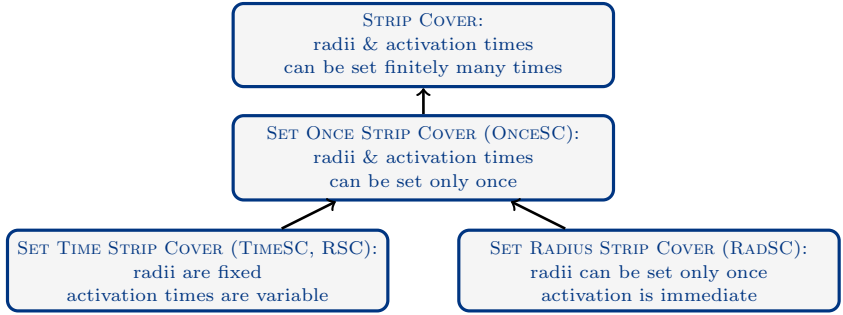
\begin{figure}[t]
\centering
		\begin{tikzpicture}[xscale=1.5,yscale=1.5
				, arrow/.style={->,draw=nym-orange,line width=1pt,bend angle=90}
				, popup/.style={rounded corners,draw=nym-blue,very thick,text width=4.8cm,text centered,fill=whitesmoke,text=nym-blue,font=\fontsize{7}{9}\selectfont}
				]
				\node[popup] (ar-rsc) at (4,4) {\strip:\\radii \& activation times\\can be set finitely many times};
				\node[popup] (set-once) at (4,3) {\sosclong (\sosc):\\radii \& activation times\\can be set only once};
				\node[popup] (rsc) at (2,2) {\stsclong (\stsc, RSC):\\radii are fixed\\activation times are variable};
				\node[popup] (all) at (6,2) {\srsclong (\srsc):\\radii can be set only once\\activation is immediate};
				\path[arrow] (rsc) edge (set-once);
				\path[arrow] (all) edge (set-once);
				\path[arrow] (set-once) edge (ar-rsc);
			\end{tikzpicture}
\caption{Relationship of Problem Variants.}
\label{fig:schematic}
\end{figure}


\paragraph*{\bf Related work.}
%
\stsc,
which is known as \textsc{Restricted Strip Covering},
was shown to be NP-hard by Buchsbaum et
al.~~\cite{buchsbaum2007restricted}, who also gave an $O(\log \log
n)$-approximation algorithm.  Later, a constant factor approximation
algorithm was discovered by Gibson and
Varadarajan~\cite{gibson2009decomposing}.


Close variants of \srsc have been the subject of previous work.
Whereas \srsc requires \emph{area} coverage, Peleg and
Lev-Tov~\cite{lev2005polynomial} studied \emph{target} coverage.  In
this problem the input is a set of $n$ sensors and a finite set of $m$
points on the line that are to be covered, and the goal is to find the
radial assignments with the minimum sum of radii.  They used dynamic
programming to devise a polynomial time alorithm.  Bar-Noy et
al.~\cite{bar2009cheap} improved the running time to $O(n+m)$.
Recently, Bar-Noy et al.~\cite{BRT13} considered an extension of \srsc
in which sensors are mobile.

\strip was first considered by Bar-Noy and
Baumer~\cite{barnoy2011maximizing}, who gave a $\threehalves$ lower
bound on the performance of \rr, the algorithm in which the sensors
take turns covering the entire region (see Observation~\ref{obs:rr}),
but were only able to show a corresponding upper bound of $1.82$.
The similar \textsc{Connected Range Assignment} (CRA) problem, in
which radii are assigned to points in the plane in order to obtain a
connected disk graph, was studied by Chambers et
al.~\cite{chambers2011connecting}. They showed that the best one
circle solution to CRA also yields a $\threehalves$-approximation
guarantee, and in fact, the instance that produces their lower bound
is simply a translation of the instance used in
Observation~\ref{obs:rr}.

The notion of \emph{duty cycling} as a mean to maximize network
lifetime was also considered in the literature of discrete geometry.  In
this context, maximizing the number of covers $t$ serves as a proxy
for maximizing the actual network lifetime.
Pach~\cite{pach1986covering} began the study of decomposability of
multiple coverings.
Pach and T\'{o}th~\cite{pach2009decomposition} showed that a $t$-fold
cover of translates of a centrally-symmetric open convex polygon can
be decomposed into $\Omega(\sqrt{t})$ covers. This
was later improved to the optimal $\Omega(t)$ covers by Aloupis et
al.~\cite{aloupis2010decomposition}, while Gibson and
Varadarajan~\cite{gibson2009decomposing} showed the same result
without the centrally-symmetric restriction.

Motivated by prior invocations of duty
cycling~\cite{slijepcevic2001power,perillo2003optimal,abrams2004set,cardei2005improving,cardei2006improving,cardei2005maximum},
Bar-Noy et al.~\cite{BBR12} studied a duty cycle variant of \sosc with
unit batteries in which sensors must be grouped into shifts of size at
most $k$ that take turns covering $[0,1]$.  (\rr is the only possible
algorithm when $k=1$.)  They presented a polynomial-time algorithm for
$k=2$ and showed that the approximation ratio of this algorithm is
$\frac{35}{24}$ for $k>2$.  It was also shown that its approximation
ratio is at least $\frac{15}{11}$, for $k \geq 4$, and $\frac{6}{5}$,
for $k = 3$. A fault-tolerance model, in which smaller shifts are more
robust, was also proposed.


\paragraph*{\bf Our results.}
%
We introduce the \textsc{Set Once} model that corresponds to the case
where the scheduler does not have the ability to vary the sensor's
radius once it has been activated.  We show that \sosc is NP-hard
(Section~\ref{sec:hardness}) and that \rr is a
$\threehalves$-approximation algorithm for both \sosc and \strip
(Section~\ref{sec:rr}).  This closes a gap between the best previously
known lower and upper bounds ($\frac{3}{2}$ and $1.82$, resp.) on the
performance of this algorithm.
Our analysis of \rr is based on the following approach:  We slice an
optimal schedule into strips in which the set of active sensors is
fixed.  For each such strip we construct an instance with unit
batteries and compare the performance of \rr to the \srsc optimum of
this instance.
In Section~\ref{sec:dc} we show that the class of duty cycle
algorithms cannot improve on this $\threehalves$ guarantee.
In Section~\ref{sec:srsc}, we provide an $O(n^2 \log{n})$-time
algorithm for \srsc.  We note that the same approach would work for
the case where, for every sensor $i$, the $i$th battery is drained in
inverse proportion to $\rho_i^\alpha$, for some $\alpha>0$.

%


\section{Preliminaries}
\label{sec:prelim}

\paragraph*{\bf Problems.}
The \textsc{Set Once Strip Cover} (abbreviated \sosc) is defined as
follows. Let $U \eqdf [0,1]$ be the interval that we wish to cover.
Given is a vector $x = (x_1,\ldots,x_n) \in U^n$ of $n$ sensor
locations, and a corresponding vector $b = (b_1,\ldots,b_n) \in
\mathbb{Q}_+^n$ of battery charges, with $b_i \geq 0$ for all $i$.  We
assume that $x_i \leq x_{i+1}$ for every $i \in \{1,\ldots,n-1\}$.  We
sometimes abuse notation by treating $x$ as a set.  An instance of the
problem thus consists of a pair $I = (x,b)$, and a solution is an
assignment of radii and activation times to sensors.  More
specifically a solution (or \emph{schedule}) is a pair $S =
(\rho,\tau)$ where $\rho_i$ is the \emph{radius} of sensor $i$ and
$\tau_i$ is the \emph{activation time} of $i$. Since the radius of
each sensor cannot be reset, this means that sensor $i$ becomes active
at time $\tau_i$, covers the \emph{range} $[x_i - \rho_i, x_i +
\rho_i]$ for $b_i/\rho_i$ time units, and then becomes inactive since
it has exhausted its entire battery.

Any schedule can be visualized by a space-time diagram in which each
coverage assignment can be represented by a rectangle. It is customary
in such diagrams to view the sensor locations as forming the
horizontal axis, with time extending upwards vertically. In this case,
the coverage of a sensor located at $x_i$ and assigned the radius
$\rho_i$ beginning at time $\tau_i$ is depicted by a rectangle with
lower-left corner $(x_i - \rho_i, \tau_i)$ and upper-right corner
$(x_i + \rho_i, \tau_i + b_i/\rho_i)$. Let the set of all points
contained in this rectangle be denoted as $Rect(\rho_i, \tau_i)$.
%
A point $(u,t)$ in space-time is \emph{covered} by a schedule $(\rho,
\tau)$ if $(u,t) \in \bigcup_{i} Rect(\rho_i, \tau_i)$. The
\emph{lifetime} of the network in a solution $S = (\rho,\tau)$ is the
maximum value $T$ such that every point $(u,t) \in U \times [0,T]$ is
covered. Graphical depictions of two schedules are shown below in
Figure~\ref{fig:rr}.

In \sosc our goal is to find a schedule $S=(\rho,\tau)$ that maximizes
the lifetime $T$. Given an instance $I = (x,b)$, the optimal lifetime
is denoted by $\opt(x,b)$.  (We sometimes use $\opt$, when the
instance is clear from the context.)

The \srsclong (\srsc) problem is a variant of \sosc in which
$\tau_i=0$, for every $i$.  Hence, a solution is simply a radial
assignment $\rho$.  \stsclong (\stsc) is another variant in which the
radii are given in the input, and a solution is an assignment of
activation times to sensors.

\strip is a generalization of \sosc in which a sensor's radius may be 
changed finitely many times.  In this case a solution is a vector of
piece-wise constant functions $\rho(t)$, where $\rho_i(t)$ is the
sensing radius of sensor $i$ at time $t$. The solution is feasible if
$U$ is covered for all $t \in [0,T]$, and if $\int_0^\infty \rho_i(t)
\, dt \leq b_i$, for every $i$.
The segment $[0,1]$ is covered at time $t$, if 
\(
[0,1] \subseteq \bigcup_i [x_i-\rho_i(t),x_i+\rho_i(t)] 
\).


\paragraph*{\bf Maximum lifetime.}
The best possible lifetime of an instance $(x,b)$ is $2\sum_i b_i$.
We state this formally for \sosc, but the same holds for the other
variants.

\begin{observation}
\label{obs:lifetime}
The lifetime of a
\sosc instance $(x,b)$ is at most $2\sum_i b_i$.
\end{observation}
\begin{proof}
Consider an optimal solution $(\rho,\tau)$ for $(x,b)$ with lifetime
$T$.  A sensor $i$ covers an interval of length $2\rho_i$ for
$\frac{b_i}{\rho_i}$ time.  The lifetime $T$
is at most the total area of space-time covered by the sensors, which
is at most $\sum_i 2 \rho_i \cdot b_i/\rho_i$.
\end{proof}


\paragraph*{\bf Round Robin.} 
%
We focus on a simple algorithm we call \rr.  The \rr algorithm forces
the sensors to take turns covering $U$, namely it assigns, for every
$i$, $\rho_i = r_i \eqdf \max \{x_i,1-x_i\}$ and $\tau_i =
\sum_{j=1}^{i-1} b_j/\rho_j$.  
The lifetime of \rr is thus
\[\textstyle
\RR(x,b) \eqdf \sum_{i=i}^n b_i/r_i
~.
\]
Notice that Observation~\ref{obs:lifetime} implies an upper bound of
$2$ on the approximation ratio of \rr, since $r_i \leq 1$, for every
$i$.
%
A lower bound of $\threehalves$ on the approximation guarantee of \rr
was given in~\cite{barnoy2011maximizing} using the two sensor instance
$x = (\frac{1}{4},\frac{3}{4})$, $b = (1,1)$.  The relevant schedules
are depicted graphically in Figure~\ref{fig:rr}.

\begin{observation}[\cite{barnoy2011maximizing}]
\label{obs:rr}
The approximation ratio of \rr is at least $\frac{3}{2}$.
\end{observation}

\begin{figure}[t]
\centering
\subfloat[$\opt(x,b) = 4$]{%
			\begin{tikzpicture}[xscale=3.25, yscale=1,
				coverage/.style={circle,draw=blue!50,fill=blue!20,thick,opacity=0.5},
				lifetime/.style={->,dashed,thick,gray},
				sensor/.style={circle,draw=black,fill=red!50,thick,opacity=0.75,inner sep=0pt,minimum size=2mm}]
				\draw[step=.5cm,gray,very thin] (0,0) grid (1, 4);	
				\draw[-] (0,0) -- (1,0) coordinate (x axis);
				\foreach \x/\xtext in {0, 0.25/\frac{1}{4}, 0.5/\frac{1}{2}, 0.75/\frac{3}{4}, 1}
					\draw (\x,3.5pt) -- (\x,-3.5pt) node[anchor=north] {$\xtext$};
				\draw[-] (0,0) -- (0,4) coordinate (y axis);
				\foreach \y/\ytext in {0, 1, 2, 3, 4}
					\draw (1pt,\y) -- (-1pt,\y) node[anchor=east] {$\ytext$};
				\coordinate (a) at (0,0);	
				\coordinate (b) at (1,0);	
				\coordinate [label=below left:$x_1$] (s1) at (1/4, 0);
				\coordinate [label=below right:$x_2$] (s2) at (3/4, 0);
				\draw [coverage] (0,0) rectangle (1/2, 4);
				\draw [coverage] (1/2,0) rectangle (1,4);
				\node [sensor] at (s1) [] {};
				\node [sensor] at (s2) [] {};
				\draw [lifetime] (s1) -- (1/4,4);
				\draw [lifetime] (s2) -- (3/4,4);
			\end{tikzpicture}
} 
\hspace{30pt}
\subfloat[$\rr(x,b) = 2 \frac{2}{3}$]{%
			\begin{tikzpicture}[xscale=3.25, yscale=1,
				coverage/.style={circle,draw=blue!50,fill=blue!20,thick,opacity=0.5},
				lifetime/.style={->,dashed,thick,gray},
				sensor/.style={circle,draw=black,fill=red!50,thick,opacity=0.75,inner sep=0pt,minimum size=2mm}]
				\draw[step=.5cm,gray,very thin] (0,0) grid (1, 4);	
				\draw[-] (0,0) -- (1,0) coordinate (x axis);
				\foreach \x/\xtext in {0, 0.25/\frac{1}{4}, 0.5/\frac{1}{2}, 0.75/\frac{3}{4}, 1}
					\draw (\x,3pt) -- (\x,-3pt) node[anchor=north] {$\xtext$};
				\draw[-] (0,0) -- (0,4) coordinate (y axis);
				\foreach \y/\ytext in {0, 1, 2, 3, 4}
					\draw (1pt,\y) -- (-1pt,\y) node[anchor=east] {$\ytext$};
				\coordinate (a) at (0,0);	
				\coordinate (b) at (1,0);	
				\coordinate [label=below left:$x_1$] (s1) at (1/4, 0);
				\coordinate [label=below right:$x_2$] (s2) at (3/4, 0);
				\draw [coverage] (-1/2,0) rectangle (1, 4/3);
				\draw [coverage] (0,4/3) rectangle (3/2,8/3);
				\node [sensor] at (s1) [] {};
				\node [sensor] at (s2) [] {};
				\draw [lifetime] (s1) -- (1/4,4/3);
				\draw [lifetime] (3/4,4/3) -- (3/4,8/3);
			\end{tikzpicture}
}
\caption{\rr vs.\ \opt with $x = (\frac{1}{4},\frac{3}{4})$ and 
$b = (1,1)$.  The sensors are indicated by (red) dots.  Each of the
(blue) rectangles represents the active coverage region for one
sensor.  The dashed gray arrow helps to clarify which sensor is active
at a particular point in time.}
\label{fig:rr}
\end{figure}
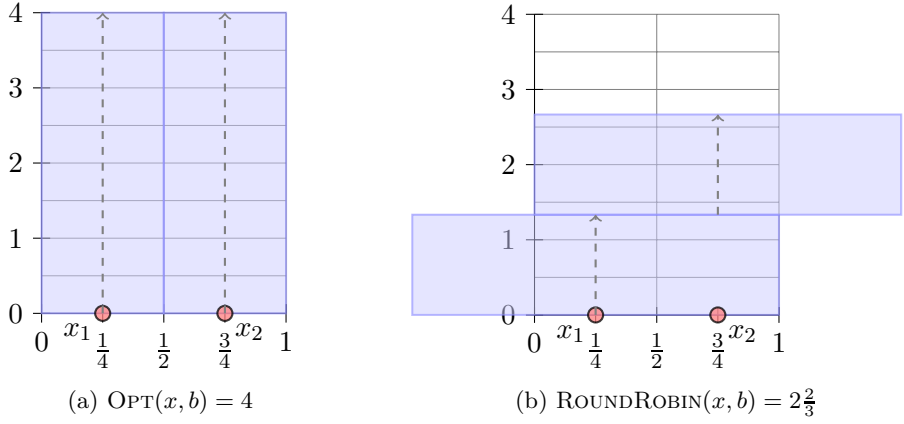	

Given an instance $(x,b)$ of \sosc, let $B \eqdf \sum_i b_i$ be the
total battery charge of the system and $\overline{r} = \sum_i
\frac{b_i}{B} \cdot r_i$ be the average of the $r_i$'s, weighted by
their respective battery charge.  We define the following lower bound
on $\RR(x,b)$:
\[
\RR'(x,b) \eqdf B/\overline{r}
~.
\]

\begin{lemma}
\label{lemma:convex}
$\RR'(x,b) \leq \RR(x,b)$, for every \sosc instance $(x,b)$.
\end{lemma}
\begin{proof}
We have that
\[
\RR(x,b) 
=    \sum_{i=1}^n \frac{b_i}{r_i} 
=    \sum_{i=1}^n \frac{b_i^2}{b_i r_i} 
\geq \frac{(\sum_{i=1}^n b_i)^2}{\sum_{i=1}^n b_i r_i} 
=    \frac{\sum_{i=1}^n b_i}{\overline{r}} 
= \RR'(x,b)
~,
\]
where the inequality is due to 
an implication of the Cauchy-Schwarz Inequality: $\sum_j
\frac{c_j^2}{d_j} \geq \frac{(\sum_j c_j)^2}{\sum_j d_j}$, for any
positive $c,d \in \reals^n$.
\end{proof}

\section{Set Once Hardness Result}
\label{sec:hardness}

In this section we show that \sosc is NP-hard.  This is done using a
reduction from \partition.
	
\begin{theorem}
\label{thm:hardness}
\sosc is NP-hard.
\end{theorem}	
\begin{proof}
Let $Y = \{y_1,\ldots,y_n\}$ be a given instance of \partition, and
define $B = \frac{1}{2} \sum_{i=1}^n y_i$.  We create an instance of
\sosc by placing $n$ sensors with battery $y_i$ at $\half$, and two
additional sensors equipped with battery $B$ at $\inv{6}$ and
$\frac{5}{6}$, respectively.  That is, the instance of \sosc consists
of sensor locations $x = (\inv{6}, \underbrace{\textstyle \half,
\ldots, \half}_n, \frac{5}{6} )$ and batteries $b = (B,y_1,\ldots,y_n,
B)$.  We show that $Y \in \partition$ if and only if the maximum
possible lifetime of $8B$ is achievable for the \sosc instance.
		
First, suppose $Y \in \partition$, hence there exist two non-empty
disjoint subsets $Y_0, Y_1 \subseteq Y$, such that $Y_0 \cup Y_1 = Y$,
and $\sum_{y \in Y_0} y = B = \sum_{y \in Y_1} y$.  Schedule the
sensors in $Y_0$ to iteratively cover the region
$[\inv{3},\frac{2}{3}]$.  Since all of these sensors are located at
$\half$, this requires that each sensor's radius be set to $\inv{6}$,
i.e. $\rho_{i+1} = \inv{6}$, for every $i \in Y_0$.
Since the sum of their batteries is $B$, this region can be covered
for exactly $6B$ time units.  With the help of the additional sensors
located at $\inv{6}$ and $\frac{5}{6}$, whose radii are also set to
$\rho_1 = \rho_{n+2} = \inv{6}$, the sensors in $Y_0$ can thus cover
$[0,1]$ for $6B$ time units (see Figure~\ref{fig:hardness} for an
example).  Next, the sensors in $Y_1$ can cover $[0,1]$ for an
additional $2B$ time units, since they all require a radius of
$\rho_{i+1} = \half$, for every $i \in Y_1$.  Thus, the total lifetime
is $8B$.

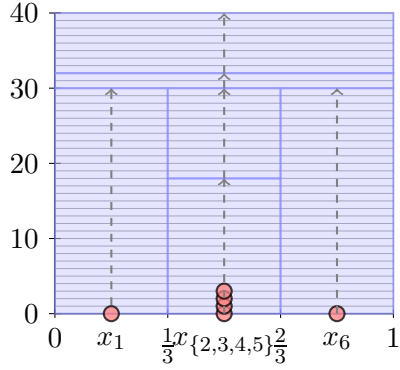
\begin{figure}[t]
\centering		
\begin{tikzpicture}[domain=0:1,xscale=4.5, yscale=0.1	
		,coverage/.style={circle,draw=blue!50,fill=blue!20,thick,opacity=0.5}
		,lifetime/.style={->,dashed,thick,gray}
		,sensor/.style={circle,draw=black,fill=red!50,thick,opacity=0.75,inner sep=0pt,minimum size=2mm}
		,moved/.style={circle,draw=black,fill=blue!50,thick,opacity=0.75,inner sep=0pt,minimum size=2mm}]
	\draw[step=1cm,gray,very thin] (0,0) grid (1, 40);	
	\draw[-] (0,0) -- (1,0) coordinate (x axis);
		\foreach \x/\xtext in {0, 0.33333/\frac{1}{3}, .666666/\frac{2}{3}, 1}
			\draw (\x, 16pt) -- (\x, -16pt) node[anchor=north] {$\xtext$};
	\draw[-] (0,0) -- (0,40) coordinate (y axis);
		\foreach \y/\ytext in {0, 10, 20, 30, 40}
			\draw (0.5pt,\y) -- (-0.5pt,\y) node[anchor=east] {$\ytext$};
	\draw [coverage] (0,0) rectangle (1/3, 30);
	\draw [coverage] (1/3,0) rectangle (2/3, 18);
	\draw [coverage] (1/3,18) rectangle (2/3, 30);
	\draw [coverage] (2/3,0) rectangle (1, 30);
	\draw [coverage] (0,30) rectangle (1, 32);
	\draw [coverage] (0,32) rectangle (1, 40);
	\node [sensor] (s1) at (1/6, 0) [label=below:$x_1$] {};
	\node [sensor] (s2) at (1/2, 0) [label=below:$x_{\{2,3,4,5\}}$] {};
	\node [sensor] (s3) at (1/2, 1) [] {};
	\node [sensor] (s4) at (1/2, 2) [] {};
	\node [sensor] (s5) at (1/2, 3) [] {};
	\node [sensor] (s6) at (5/6, 0) [label=below:$x_6$] {};
	\draw [lifetime] (s1) -- (1/6,30);
	\draw [lifetime] (1/2,0) -- (1/2,18);
	\draw [lifetime] (1/2,18) -- (1/2,30);
	\draw [lifetime] (5/6,0) -- (5/6,30);
	\draw [lifetime] (1/2,30) -- (1/2,32);
	\draw [lifetime] (1/2,32) -- (1/2,40);
\end{tikzpicture} 	
\caption{Proof of NP-hardness. $Y = \{1,2,3,4\}$ is a given instance 
of \textsc{Partition}, and $(x,b) = \left( (\inv{6}, \inv{2}, \ldots,
\inv{2}, \frac{5}{6}), (5,1,2,3,4,5) \right)$ is the translated \sosc
instance.}  
\label{fig:hardness}
\end{figure}

Now suppose that for such a \sosc instance, the lifetime of $8B$ is
achievable. Since the maximum possible lifetime is achievable, no
coverage can be wasted in the optimal schedule. In this case the radii
of the sensors at $\inv{6}$ and $\frac{5}{6}$ must be exactly
$\inv{6}$, since otherwise, they would either not reach the endpoints
$\{0,1\}$, or extend beyond them. Moreover, due the fact that all of
the other sensors are located at $\half$, and their coverage is thus
symmetric with respect to $\half$, it cannot be the case that sensor
$1$ and sensor $n+2$ are active at different times. Thus, the solution
requires a partition of the sensors located at $\half$ into two
groups: the first of which must work alongside sensors $1$ and $n+2$
with a radius of $\inv{6}$ and a combined lifetime of $6B$; and the
second of which must implement \rr for a lifetime of $2B$. The
batteries of these two partitions form a solution to \partition.
\end{proof}


\section{Round Robin}
\label{sec:rr}

We show in Appendix~\ref{sec:hardness} that \sosc is NP-hard, so 
we turn our attention to approximation algorithms.  While \rr is among
the simplest possible algorithms (note that its running time is
exactly $n$), the precise value of its approximation ratio is not
obvious (although it is not hard to see that $2$ is an upper bound).
In~\cite{barnoy2011maximizing} an upper bound of $1.82$ and a lower
bound of $\threehalves$ were shown.
In this section, we show that the approximation ratio of \rr in \sosc
is exactly $\threehalves$.
The structure of the proof is as follows.  We start with an optimal
schedule $S$, and cut it into disjoint time intervals, or strips, such
that the same set of sensors is active within each time interval.
Each strip induces a \srsc instance $I_j$ and a corresponding solution
$S_j$.
Next, we show that for any such instance $I_j$, there exists a
unit-battery instance $I_j'$ with the same optimum lifetime. Finally,
we prove a lower bound on the performance of \rr on such unit battery
instances. By combining these results, we prove that $\RR(x,b) \geq
\frac{2}{3} T$.
%


\subsection{Cutting the Schedule into Strips}

Given an instance $I = (x,b)$, and a solution $S = (\rho,\tau)$ with
lifetime $T$, let $\Omega$ be the set of times until $T$ in which a
sensor was turned on or off, namely
\(
\Omega = \bigcup_i \{\tau_i,\tau_i+b_i/\rho_i\} \cap [0,T]
\).
Let $\Omega = \set{0=\omega_0,\ldots,\omega_\ell = T}$, where $\omega_j <
\omega_{j+1}$, for every $j$. 
Next, we partition the time interval $[0,T]$ into the sub-intervals
$[\omega_j,\omega_{j+1}]$, for every $j \in \{0,\ldots,\ell-1\}$.  

Next, we define a new instance for every sub-interval.  For every $j
\in \{0,\ldots,\ell-1\}$, let $x^j \subseteq x$ be the set of sensors
that participate in covering $[0,1]$ during the $j$th sub-interval of
time, i.e.,
\(
x^j = 
\set{x_i : [\omega_j,\omega_{j+1}] \subseteq [\tau_i, \tau_i+b_i/ \rho_i] }
\).
Also, let $T_j = \omega_{j+1} - \omega_j$, and let $b^j_i$ be the
energy that was consumed by sensor $i$ during the $j$th sub-interval,
i.e., $b^j_i = \rho_i \cdot T_j$.  Observe that $I_j = (x^j,b^j)$ is a
valid instance of \srsc, for which $\rho^j$, where $\rho^j_i =
\rho_i$ for every sensor $i$ such that $x_i \in x^j$, is a solution
that achieves a lifetime of exactly $T_j$.
Figure \ref{fig:strip_cut} 
provides an illustration of this procedure.

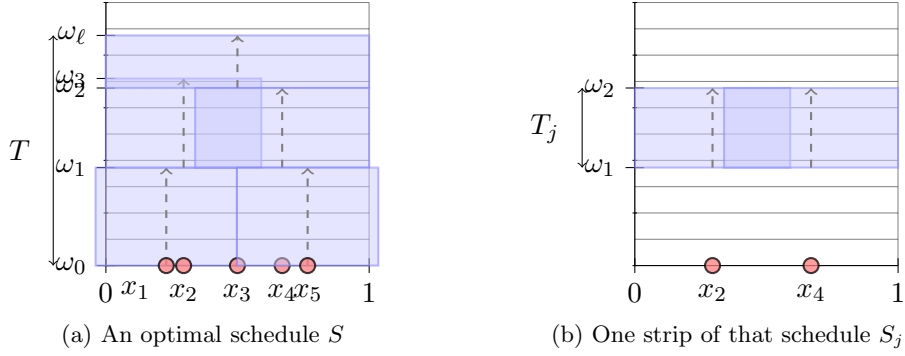
\begin{figure}[t]
\centering		
\subfloat[An optimal schedule $S$]{%
\begin{tikzpicture}[domain=0:1,xscale=3.5, yscale=0.35
		,coverage/.style={circle,draw=blue!50,fill=blue!20,thick,opacity=0.5}
		,lifetime/.style={->,dashed,thick,gray}
		,sensor/.style={circle,draw=black,fill=red!50,thick,opacity=0.75,inner sep=0pt,minimum size=2mm}
		,moved/.style={circle,draw=black,fill=blue!50,thick,opacity=0.75,inner sep=0pt,minimum size=2mm}]
	\def\maxT{10}
	\def\T{8.749078}
	\def\xa{0.2293328}
	\def\xb{0.2951739}
	\def\xc{0.499632}
	\def\xd{0.5197192}
	\def\xe{0.5370125}
	\def\xf{0.6695036}
	\def\xg{0.7662748}
	\draw[step=1cm,gray,very thin] (0,0) grid (1, \maxT);	
	\draw[-] (0,0) -- (1,0) coordinate (x axis);
		\foreach \x/\xtext in {0, 1}
			\draw (\x, 8pt) -- (\x, -8pt) node[anchor=north] {$\xtext$};
	\draw[-] (0,0) -- (0,\maxT) coordinate (y axis);
		\foreach \y/\ytext in {0, 2, ..., \maxT}
			\draw (0.25pt,\y) -- (-0.25pt,\y) node[anchor=east] {};	
	\def\wa{3.724797}
	\def\wb{6.750549}
	\def\wc{7.112631}
	\draw (1pt,0) -- (-1pt,0) node[anchor=east] {$\omega_0$};
	\draw (1pt,\wa) -- (-1pt,\wa) node[anchor=east] {$\omega_1$};
	\draw (1pt,\wb) -- (-1pt,\wb) node[anchor=east] {$\omega_2$};
	\draw (1pt,\wc) -- (-1pt,\wc) node[anchor=east] {$\omega_3$};
	\draw (1pt,\T) -- (-1pt,\T) node[anchor=east] {$\omega_\ell$};
	\node [] at (-0.21, 0.5*\T) [label=left:$T$] {};
	\draw [<->] (-0.2,0) -- (-0.2,\T);
	\draw [coverage] (-0.03913823, 0) rectangle (0.4978038, 3.724797);
	\draw [lifetime] (0.2293328, 0) -- (0.2293328, 3.724797);
	\node [sensor] (s1) at (0.2293328, 0) [label=below left:$x_{1}$] {};
	\draw [coverage] (0, 3.724797) rectangle (0.5903477, 7.112631);
	\draw [lifetime] (0.2951739, 3.724797) -- (0.2951739, 7.112631);
	\node [sensor] (s2) at (0.2951739, 0) [label=below:$x_{2}$] {};
	\draw [coverage] (-0.0007360573, 6.750549) rectangle (1, 8.749078);
	\draw [lifetime] (0.499632, 6.750549) -- (0.499632, 8.749078);
	\node [sensor] (s3) at (0.499632, 0) [label=below:$x_{3}$] {};
	\draw [coverage] (0.3390073, 3.724797) rectangle (1, 6.750549);
	\draw [lifetime] (0.6695036, 3.724797) -- (0.6695036, 6.750549);
	\node [sensor] (s6) at (0.6695036, 0) [label=below:$x_{4}$] {};
	\draw [coverage] (0.4978038, 0) rectangle (1.034746, 3.724797);
	\draw [lifetime] (0.7662748, 0) -- (0.7662748, 3.724797);
	\node [sensor] (s7) at (0.7662748, 0) [label=below:$x_{5}$] {};
\end{tikzpicture}

%
%
}
\hspace{40pt}
\subfloat[One strip of that schedule $S_j$]{%
\begin{tikzpicture}[domain=0:1,xscale=3.5, yscale=0.35
		,coverage/.style={circle,draw=blue!50,fill=blue!20,thick,opacity=0.5}
		,lifetime/.style={->,dashed,thick,gray}
		,sensor/.style={circle,draw=black,fill=red!50,thick,opacity=0.75,inner sep=0pt,minimum size=2mm}
		,moved/.style={circle,draw=black,fill=blue!50,thick,opacity=0.75,inner sep=0pt,minimum size=2mm}]
	\def\maxT{10}
	\def\T{12.53535}
	\def\xa{0.2293328}
	\def\xb{0.2951739}
	\def\xc{0.499632}
	\def\xd{0.5197192}
	\def\xe{0.5370125}
	\def\xf{0.6695036}
	\def\xg{0.7662748}
	\draw[step=1cm,gray,very thin] (0,0) grid (1, \maxT);	
	\draw[-] (0,0) -- (1,0) coordinate (x axis);
		\foreach \x/\xtext in {0, 1}
			\draw (\x, 8pt) -- (\x, -8pt) node[anchor=north] {$\xtext$};
	\draw[-] (0,0) -- (0,\maxT) coordinate (y axis);
		\foreach \y/\ytext in {0, 2, ..., \maxT}
			\draw (0.25pt,\y) -- (-0.25pt,\y) node[anchor=east] {};
	\def\wa{3.724797}
	\def\wb{6.750549}
	\draw (1pt,\wa) -- (-1pt,\wa) node[anchor=east] {$\omega_1$};
	\draw (1pt,\wb) -- (-1pt,\wb) node[anchor=east] {$\omega_2$};
	\node [] at (-0.21, 5.2) [label=left:$T_j$] {};
	\draw [<->] (-0.2,3.724797) -- (-0.2,6.750549);
	\draw [coverage] (0, 3.724797) rectangle (0.5903477, 6.750549);
	\draw [lifetime] (0.2951739, 3.724797) -- (0.2951739, 6.750549);
	\node [sensor] (s2) at (0.2951739, 0) [label=below:$x_{2}$] {};
	\draw [coverage] (0.3390073, 3.724797) rectangle (1, 6.750549);
	\draw [lifetime] (0.6695036, 3.724797) -- (0.6695036, 6.750549);
	\node [sensor] (s6) at (0.6695036, 0) [label=below:$x_{4}$] {};
\end{tikzpicture}

%
%
}
\caption{Cutting an optimal schedule into strips. Note that coverage 
overlaps may occur in both the horizontal and vertical directions in
the optimal schedule, but only horizontally in a strip.}
\label{fig:strip_cut}
\end{figure}

We further modify the instance $I_j=(x^j,b^j)$ and the solution
$\rho^j$ as follows:
\begin{itemize*}
\item Starting with $i=1$, remove sensor $i$ from the instance, if 
      the interval $[0,1]$ is covered during $[\omega_j,\omega_{j+1}]$
      without $i$.
\item Decrease the battery and the radius of the left-most sensor as 
      much as possible, and also decrease the battery and the radius
      of the right-most sensor as much as possible.
\end{itemize*}

\begin{observation}
\label{obs:annoying}
Let sensors $1$ and $m$ be the leftmost and rightmost sensors in
$x^j$.  Then, either $\rho^j_1 = x^j_1$ or the interval
$[0,x^j_1+\rho^j_1)$ is only covered by sensor $1$.  Also, either
$\rho^j_m = 1-x^j_m$ or the interval $(x^j_m-\rho^j_m,1]$ is only
covered by sensor $m$.
\end{observation}


For now, it is important to note only that $\RR(x^j,b^j) = \sum_{x_i
\in x^j} \frac{b^j_i}{r_i}$ is the \rr lifetime of the $j$th strip,
which is a specific \srsc instance $I_j$ with the properties outlined
above.

\subsection{Reduction to Set Radius Strip Cover with Uniform Batteries}

Given the \srsc instance $I_j = (x^j,b^j)$ and a solution $\rho^j$, we
construct an instance $I'_j = (y^j,\mathbf{1})$ with unit size
batteries and a \srsc solution $\sigma^j$, such that the lifetime of
$\sigma^j$ is $T_j$. That is, the optimal lifetime of $I_j'$ is
exactly the same as for $I_j$, but it uses only unit batteries.

Let $\opt_0$ denote the optimal \srsc lifetime.  We assume that $b^j_i
\in \naturals$ and $b^j_i \geq 3$ for every $i$, since
\begin{inparaenum}[(i)]
\item $b^j_i \in \rationals$ for every $i$,
\item $\opt_0(x,\beta b) = \beta \cdot \opt_0(x,b)$, and
\item $\RR(x, \beta b) = \beta \cdot  \RR(x,b)$.
\end{inparaenum}

The instance $I'_j$ is constructed as follows.  We replace each sensor
$i$ such that $x^j_i \in x^j$ with $b^j_i$ unit battery sensors whose
average location is $x_i$.  These unit battery sensors are called the
\emph{children} of $i$.  To do this, we divide the interval $[x^j_i -
\rho^j_i, x^j_i + \rho^j_i]$ into $b^j_i$ equal sub-intervals, and
place a unit battery sensor in the middle of each sub-interval.
Observe that child sensors may be placed outside $[0,1]$, namely to
the left of $0$ or to the right of $1$.
The solution $\sigma^j$ is defined as follows.  For any child $k$ of a
sensor $i$ in $I_j$, we set $\sigma^j_k = \rho^j_i/b^j_i$.  An example
is shown in Figure~\ref{fig:strip_split}.

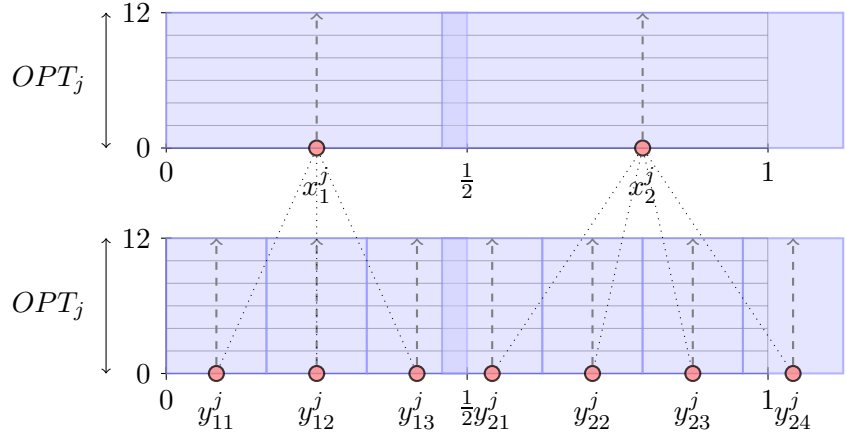
\begin{figure}[t]
\centering		
\begin{tikzpicture}[domain=0:1,xscale=8, yscale=0.3
		,coverage/.style={circle,draw=blue!50,fill=blue!20,thick,opacity=0.5}
		,lifetime/.style={->,dashed,thick,gray}
		,sensor/.style={circle,draw=black,fill=red!50,thick,opacity=0.75,inner sep=0pt,minimum size=2mm}
		,moved/.style={circle,draw=black,fill=blue!50,thick,opacity=0.75,inner sep=0pt,minimum size=2mm}]
	\draw[step=1cm,gray,very thin] (0,0) grid (1, 6);	
	\draw[step=1cm,gray,very thin] (0,10) grid (1, 16);	
	\draw[-] (0,0) -- (1,0) coordinate (x axis);
		\foreach \x/\xtext in {0, .5/\frac{1}{2}, 1}
			\draw (\x, 8pt) -- (\x, -8pt) node[anchor=north] {$\xtext$};
	\draw[-] (0,10) -- (1,10) coordinate (x axis);
		\foreach \x/\xtext in {0, .5/\frac{1}{2}, 1}
			\draw (\x, 10.2) -- (\x, 9.8) node[anchor=north] {$\xtext$};
	\draw[-] (0,0) -- (0,6) coordinate (y axis);
		\foreach \y/\ytext in {0, 6/12, 10/0, 16/12}
			\draw (0.25pt,\y) -- (-0.25pt,\y) node[anchor=east] {$\ytext$};
	\node [] at (-0.1, 3) [label=left:$OPT_j$] {};
	\node [] at (-0.1, 13) [label=left:$OPT_j$] {};
	\draw [<->] (-0.1,0) -- (-0.1,6);
	\draw [<->] (-0.1,10) -- (-0.1,16);
	\draw [coverage] (0,0) rectangle (1/6, 6);
	\draw [coverage] (1/6,0) rectangle (2/6, 6);
	\draw [coverage] (2/6,0) rectangle (3/6, 6);
	\draw [coverage] (11/24,0) rectangle (15/24, 6);
	\draw [coverage] (15/24,0) rectangle (19/24, 6);
	\draw [coverage] (19/24,0) rectangle (23/24, 6);
	\draw [coverage] (23/24,0) rectangle (27/24, 6);
	\draw [coverage] (0,10) rectangle (1/2, 16);
	\draw [coverage] (11/24,10) rectangle (27/24, 16);
	\node [sensor] (si11) at (1/12, 0) [label=below:$y_{11}^j$] {};
	\node [sensor] (si12) at (3/12, 0) [label=below:$y_{12}^j$] {};
	\node [sensor] (si13) at (5/12, 0) [label=below:$y_{13}^j$] {};
	\node [sensor] (si21) at (13/24, 0) [label=below:$y_{21}^j$] {};
	\node [sensor] (si22) at (17/24, 0) [label=below:$y_{22}^j$] {};
	\node [sensor] (si23) at (21/24, 0) [label=below:$y_{23}^j$] {};
	\node [sensor] (si24) at (25/24, 0) [label=below:$y_{24}^j$] {};
	\draw [lifetime] (si11) -- (1/12,6);
	\draw [lifetime] (si12) -- (3/12,6);
	\draw [lifetime] (si13) -- (5/12,6);
	\draw [lifetime] (si21) -- (13/24,6);
	\draw [lifetime] (si22) -- (17/24,6);
	\draw [lifetime] (si23) -- (21/24,6);
	\draw [lifetime] (si24) -- (25/24,6);
	\node [sensor] (si1) at (1/4, 10) [label=below:$x_{1}^j$] {};
	\node [sensor] (si2) at (19/24, 10) [label=below:$x_{2}^j$] {};
	\draw [lifetime] (si1) -- (1/4,16);
	\draw [lifetime] (si2) -- (19/24,16);
	\draw [dotted] (si1) -- (si11);
	\draw [dotted] (si1) -- (si12);
	\draw [dotted] (si1) -- (si13);
	\draw [dotted] (si2) -- (si21);
	\draw [dotted] (si2) -- (si22);
	\draw [dotted] (si2) -- (si23);
	\draw [dotted] (si2) -- (si24);
\end{tikzpicture}
\caption{Reduction of a non-uniform battery strip $I_j$ to a uniform 
battery instance $I_j'$: At the top, $I_j = \left( (\frac{1}{4},
\frac{19}{24}), (3,4) \right)$, while at the bottom, $I_j' = \left(
(\frac{1}{12}, \frac{3}{12}, \frac{5}{12}, \frac{13}{24},
\frac{17}{24}, \frac{21}{24}, \frac{25}{24} ), \mathbf{1}
\right)$.}
\label{fig:strip_split}
\end{figure}


\begin{lemma}
\label{lemma:unit-sol}
The lifetime of $\sigma^j$ is $T_j$.
\end{lemma}
\begin{proof}
First, the $b^j_i$ children of a sensor $i$ in $I_j$ cover the
interval $[x^j_i-\rho^j_i, x^j_i + \rho^j_i]$.  Also, a child $k$ of
$i$ survives $1/\sigma^j_k = b^j_i/\rho^j_i = T_j$ time units.
\end{proof}

Next, we prove that the lower bound on the performance of \rr may only
decrease.

\begin{lemma}
\label{lemma:unit-convex}
$\RR'(y^j,\mathbf{1}) \leq \RR'(x^j,b^j)$.
\end{lemma}
\begin{proof}
Let $p^j$ be the \rr radii of $y^j$.
Observe that if $x^j_i \leq \half$, it follows that
\[
\sum_{k:k \text{ correspond to } i} \!\!\!\!\!\! p^j_k 
=    \sum_{k:k \text{ correspond to } i} \!\!\!\!\!\! \max\set{y^j_k,1-y^j_k}
\geq \sum_{k:k \text{ correspond to } i} \!\!\!\!\!\! (1-y^j_k)
~=~  b^j_i (1-x^j_i)
~=~  b^j_i r^j_i
~,
\]
and that if $x^j_i \geq \half$, we have that
\[
\sum_{k:k \text{ correspond to } i} \!\!\!\!\!\! p^j_k 
=    \sum_{k:k \text{ correspond to } i} \!\!\!\!\!\! \max\set{y^j_k,1-y^j_k}
\geq \sum_{k:k \text{ correspond to } i} \!\!\!\!\!\! y^j_k
~=~  b^j_i x^j_i
~=~  b^j_i r^j_i
~.
\]
Hence, 
\[
\RR'(y^j, \mathbf{1})
=    \frac{\sum_i b^j_i}{\overline{p^j}}
=    \frac{B^j}{\frac{1}{B^j} \sum_{k} p_k^j}
\leq \frac{B^j}{\frac{1}{B^j} \sum_{i} b^j_i r^j_i}
=    \frac{B^j}{\overline{r^j}}
=    \RR'(x^j, b^j) 
~,
\]
and the lemma follows.
\end{proof}


\subsection{Analysis of Round Robin for Unit Batteries}

For the remainder of this section, we assume that we are given a unit
battery instance $x$ that corresponds to the $j$th strip.  (We drop
the subscript $j$ and go back to $x$ for readability.)  Recall that $x
\cap [0,1]$ is not necessarily equal to $x$, since some children 
could have been created outside $[0,1]$ in the previous step.  We show
that $\RR'(x) \geq \frac{2}{3} \opt_0(x)$.

Let $i_0 = \min \set{i : x_i \geq 0}$ and let $i_1 = \max \set{i : x_i
\leq 1}$ be the indices of the left-most and right-most sensors in
$[0,1]$, respectively.  

\begin{lemma}
\label{lemma:Delta}
$\max_{i \in \{i_0,\ldots,i_1-1\}} \{x_{i+1}-x_i\}
= \max_{i \in \{1,\ldots,n-1\}} \{x_{i+1}-x_i\}$.
\end{lemma}
\begin{proof}
By Observation~\ref{obs:annoying} either $\rho_1 = x_1$ and hence none
of its children are located to the left of $0$, or the points to the
left of $x_1+\rho_1$ are only covered by sensor $1$, which means that
the gaps between $1$'s children to the left of zero also appears
between its children within $[0,1]$.  (Recall that $b^j_i \geq 3$, for
all $i$.)  The same argument can be used for the right-most sensor.
\end{proof}

As illustrated in Figure~\ref{fig:delta_definition}, we define
\begin{align*}
\Delta_0 &\eqdf 
\begin{cases}
x_{i_0}-x_{i_0-1} & i_0 > 1, -x_{i_0-1} < x_{i_0}, \\
2x_{i_0}          & \text{otherwise},
\end{cases}
\end{align*}
\begin{align*}
\Delta_1 & \eqdf   
\begin{cases}
x_{i_1+1}-x_{i_1} & i_1 < n, x_{i_1+1}-1 < 1-x_{i_1}, \\
2(1-x_{i_1})      & \text{otherwise,}
\end{cases}
\end{align*}
and
\[
\Delta \eqdf 
       \max \set{ \Delta_0, \Delta_1, \max_{i \in \{i_0,\ldots,i_1-1\}} \{x_{i+1}-x_i\} }
~.
\]
%

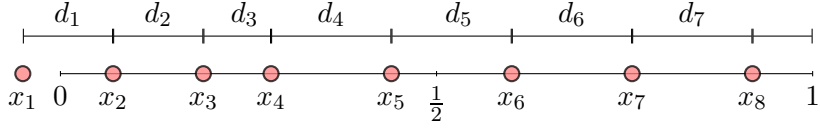
\begin{figure}[t]
\centering		
\begin{tikzpicture}[domain=0:1,xscale=10, yscale=1
,sensor/.style={circle,draw=black,fill=red!50,thick,opacity=0.75,inner sep=0pt,minimum size=2mm}
,moved/.style={circle,draw=black,fill=blue!50,thick,opacity=0.75,inner sep=0pt,minimum size=2mm}]
\draw[-] (0,0) -- (1,0) coordinate (x axis);
	\foreach \x/\xtext in {0, 0.5/\frac{1}{2}, 1}
		\draw (\x,1pt) -- (\x,-1pt) node[anchor=north] {$\xtext$};
\node [sensor] (s0) at (-0.05, 0) [label=below:$x_1$] {};
\node [sensor] (s1) at (0.07, 0) [label=below:$x_2$] {};
\node [sensor] (s2) at (0.19, 0) [label=below:$x_3$] {};
\node [sensor] (s3) at (0.28, 0) [label=below:$x_4$] {};
\node [sensor] (s4) at (0.44, 0) [label=below:$x_5$] {};
\node [sensor] (s5) at (0.6, 0) [label=below:$x_6$] {};
\node [sensor] (s6) at (0.76, 0) [label=below:$x_7$] {};
\node [sensor] (s7) at (0.92, 0) [label=below:$x_8$] {};
\draw[|-|] (-0.05, 1/2) -- (0.07, 1/2) node[above] at (0.01,1/2) {$d_1$};
\draw[|-|] (0.07, 1/2) -- (0.19, 1/2) node[above] at (0.13,1/2) {$d_2$};
\draw[|-|] (0.19, 1/2) -- (0.28, 1/2) node[above] at (0.245,1/2) {$d_3$};
\draw[|-|] (0.28, 1/2) -- (0.44, 1/2) node[above] at (0.36,1/2) {$d_4$};
\draw[|-|] (0.44, 1/2) -- (0.6, 1/2) node[above] at (0.53,1/2) {$d_5$};
\draw[|-|] (0.6, 1/2) -- (0.76, 1/2) node[above] at (0.68,1/2) {$d_6$};
\draw[|-|] (0.76, 1/2) -- (0.92, 1/2) node[above] at (0.84,1/2) {$d_7$};
\draw[|-|] (0.92, 1/2) -- (1, 1/2);
\end{tikzpicture}
\caption{Illustration of the gaps in a unit battery instance $x$. 
Note that $i_0 = 2$ and $i_1 = 8$.  $\Delta_0 = d_1$, since sensor $1$
is closer to $0$ than sensor $2$.  Also, $\Delta_1 = 2(1-x_8)$.  Hence,
$\Delta = \max\set{d_4,d_1,2(1-x_8)}$.}
\label{fig:delta_definition}
\end{figure}

We describe the optimal \srsc lifetime in terms of $\Delta$.

\begin{lemma}
\label{lemma:unit-opt}
The optimum lifetime of $x$ is $\frac{2}{\Delta}$.
\end{lemma}
\begin{proof}
To verify that $2/\Delta$ can be achieved, consider the solution in
which $\rho_i = \Delta/2$ for all $i$.  Clearly, $[0,1]$ is covered,
and all sensors die after $2/\Delta$ time units.  Now suppose that a
solution $\rho$ exists with lifetime strictly greater than $2/\Delta$.
Hence $\max_i \{\rho_i\} < \Delta /2$.  By definition, $\Delta$ must
equal $\Delta_0$, $\Delta_1$, or the maximum internal gap.  If the
latter, then there exists a point $u \in [0,1]$ between the two
sensors forming the maximum internal gap that is uncovered. On the
other hand, if $\Delta = \Delta_0$, then if $\Delta_0 = 2x_{i_0}$, 0
is uncovered, and otherwise, there is a point in $[0,x_{i_0}]$ that is
uncovered.  A similar argument holds if $\Delta = \Delta_1$.
\end{proof}

In the next definition we transform $x$ into an instance $x'$ by
pushing sensors away from $\half$, so that each internal gap between
sensors is of equal width.
See Figure~\ref{fig:stretching} for an illustration. 

\begin{definition}
For a given instance $x$, let $k$ be a sensor whose location is
closest to $1/2$.  Then we define the \emph{stretched} instance $x'$
of $x$ as follows:
\[
x_i' = 
\begin{cases} 
(1-r_k) - (\ceil{n/2}-i) \Delta & i \leq \ceil{n/2}, \\
(1-r_k) + (i-\ceil{n/2}) \Delta & i > \ceil{n/2}. \\
\end{cases}
\] 
\end{definition}	

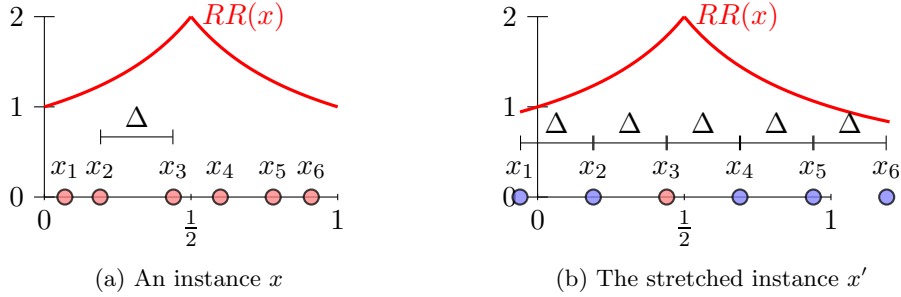
\begin{figure}[t]
\centering
\subfloat[An instance $x$]{%
\begin{tikzpicture}[domain=0:1,xscale=3.9, yscale=1.2
		,sensor/.style={circle,draw=black,fill=red!50,thick,opacity=0.75,inner sep=0pt,minimum size=2mm}
		,moved/.style={circle,draw=black,fill=blue!50,thick,opacity=0.75,inner sep=0pt,minimum size=2mm}]
	\draw[-] (0,0) -- (1,0) coordinate (x axis);
		\foreach \x/\xtext in {0, 0.5/\frac{1}{2}, 1}
			\draw (\x,1pt) -- (\x,-1pt) node[anchor=north] {$\xtext$};
	\draw[-] (0,0) -- (0,2) coordinate (y axis);
		\foreach \y/\ytext in {0, 1, 2}
			\draw (1pt,\y) -- (-1pt,\y) node[anchor=east] {$\ytext$};
	\draw[color=red,very thick,domain=0:0.5] plot (\x,{1/(1-\x)}) node[right] {${RR}(x)$};
	\draw[color=red,very thick,domain=0.5:1] plot (\x,1/\x) node[right] {};
	\def\a{0.44}
	\def\d{0.25}
	\def\dd{0.125}
	\node [sensor] (s1) at (0.07, 0) [label=above:$x_1$] {};
	\node [sensor] (s2) at (0.19, 0) [label=above:$x_2$] {};
	\node [sensor] (s3) at (\a, 0) [label=above:$x_3$] {};
	\node [sensor] (s4) at (0.6, 0) [label=above:$x_4$] {};
	\node [sensor] (s5) at (0.78, 0) [label=above:$x_5$] {};
	\node [sensor] (s6) at (0.91, 0) [label=above:$x_6$] {};
	\draw[|-|] (\a, 2/3) -- (\a - \d, 2/3) node[above] at (\a - \dd,2/3) {$\Delta$};
\end{tikzpicture} 	
}
\hspace{40pt}
\subfloat[The stretched instance $x'$]{%
\begin{tikzpicture}[domain=0:1,xscale=3.9, yscale=1.2
		,sensor/.style={circle,draw=black,fill=red!50,thick,opacity=0.75,inner sep=0pt,minimum size=2mm}
		,moved/.style={circle,draw=black,fill=blue!50,thick,opacity=0.75,inner sep=0pt,minimum size=2mm}]
	\draw[-] (0,0) -- (1,0) coordinate (x axis);
		\foreach \x/\xtext in {0, 0.5/\frac{1}{2}, 1}
			\draw (\x,1pt) -- (\x,-1pt) node[anchor=north] {$\xtext$};
	\draw[-] (0,0) -- (0,2) coordinate (y axis);
		\foreach \y/\ytext in {0, 1, 2}
			\draw (1pt,\y) -- (-1pt,\y) node[anchor=east] {$\ytext$};
	\draw[color=red,very thick,domain=-0.06:0.5] plot (\x,{1/(1-\x)}) node[right] {${RR}(x)$};
	\draw[color=red,very thick,domain=0.5:1.2] plot (\x,1/\x) node[right] {};
	\def\a{0.44}
	\def\d{0.25}
	\def\dd{0.125}
	\def\h{0.6}
	\node [moved] (s1) at (\a - 2*\d, 0) [label=above:$x_1$] {};
	\node [moved] (s2) at (\a - \d, 0) [label=above:$x_2$] {};
	\node [sensor] (s3) at (\a, 0) [label=above:$x_3$] {};
	\node [moved] (s4) at (\a + \d, 0) [label=above:$x_4$] {};
	\node [moved] (s5) at (\a + 2*\d, 0) [label=above:$x_5$] {};
	\node [moved] (s6) at (\a + 3*\d, 0) [label=above:$x_6$] {};
	\draw[|-|] (\a - 2*\d, \h) -- (\a - \d, \h) node[above] at (\a - 3*\dd, \h) {$\Delta$};
	\draw[|-|] (\a, \h) -- (\a - \d, \h) node[above] at (\a - \dd, \h) {$\Delta$};
	\draw[|-|] (\a, \h) -- (\a + \d, \h) node[above] at (\a + \dd, \h) {$\Delta$};
	\draw[|-|] (\a + \d, \h) -- (\a + 2*\d, \h) node[above] at (\a + 3*\dd, \h) {$\Delta$};
	\draw[|-|] (\a + 2*\d, \h) -- (\a + 3*\d, \h) node[above] at (\a + 5*\dd, \h) {$\Delta$};
\end{tikzpicture} 	
}
\vspace{-5pt}
\caption{Transformation of instance $x$ to stretched instance $x'$. The sensor closest to $\half$ ($x_3$) remains in place, while the other sensors are placed at increasing intervals of $\Delta$ away from $x_3$. The \rr lifetime of a sensor is shown as a continuous function of its location $x$.}
\label{fig:stretching}
\end{figure}
	

%

\begin{observation}
Let $x'$ be a stretched instance of $x$.  Then $|\set{i : x'_i \leq
\half}| = \ceil{\frac{n}{2}}$ and $|\set{i : x'_i > \half}| =
\floor{\frac{n}{2}}$.
\end{observation}

\begin{lemma}
\label{lemma:stretch}
Let $x'$ be the stretched instance of $x$.  Then, $\opt_0(x') =
\opt_0(x)$ and $\RR'(x') \leq \RR'(x)$.
\end{lemma}
\begin{proof}
First, by construction, the internal gaps in $x'$ are of length
$\Delta$ and $\Delta'_0,\Delta'_1 \leq \Delta$.  Thus, by
Lemma~\ref{lemma:unit-opt}, $\opt_0(x') = \opt_0(x)$.
By Lemma~\ref{lemma:Delta} we know that the sensors moved away from
$\half$, hence $\sum_i r'_i \geq \sum_i r_i$ and $\RR'(x') \leq
\RR'(x)$.
\end{proof}

Now we are ready to bound $\RR(x)$.
	
\begin{lemma}
\label{lemma:ratio}
$\RR'(x) \geq \frac{2}{3} \opt_0(x)$, for every instance $I =
(x,\mathbf{1})$ of \srsc, where sensors may be located outside
$[0,1]$.
\end{lemma}
\begin{proof}
By Lemma~\ref{lemma:stretch} we may assume that the instance is
stretched.  
First, suppose that $n$ is even.  Since $x$ is a stretched instance,
it must be the case that exactly half of the sensors lie to the left
of $1/2$, and exactly half lie to the right.
Hence,
\begin{align*}
\overline{r} 
\eqdf    \inv{n} \sum_{i=1}^n r_i 
& =    \inv{n} 
         \left[ \sum_{j=0}^{n/2-1} (r_{n/2} + j \Delta)  + 
                \sum_{j=0}^{n/2-1} (r_{n/2+1} + j \Delta) 
         \right] \\
& = \inv{n} 
       \left[ \frac{n}{2} \cdot r_{n/2} + \Delta \binom{n/2}{2} + 
              \frac{n}{2} \cdot r_{n/2+1} + \Delta \binom{n/2}{2} 
       \right] \\
& = \frac{r_{n/2} + r_{n/2+1}}{2} + \frac{2 \Delta}{n} \binom{n/2}{2} \\
& = \frac{1 + \Delta}{2} + \frac{\Delta (n-2)}{4} \\
& = \frac{1}{2} + \frac{n\Delta}{4}
~,
\end{align*}
where we have used the fact that since the sequence is stretched
$r_{n/2} + r_{n/2+1} = 1 + \Delta$.  Furthermore, since $n \Delta
\geq 1$, it now follows that
\[
\frac{\RR'(x)}{\opt_0(x)} 
=    \frac{n/\overline{r}}{2/\Delta} 
=	 \frac{n \Delta}{1 + n \Delta/2} 
=    \frac{1}{\frac{1}{n\Delta} + \frac{1}{2}} 
\geq \frac{2}{3} ~.
\]
		
If $n$ is odd, then w.l.o.g.\ 
there are $\frac{n+1}{2}$ sensors to the left of $1/2$, and
$\frac{n-1}{2}$ to the right.  Then
\begin{align*}
\overline{r} 
& = \inv{n} 
       \left[ \sum_{j=0}^{(n-1)/2} (r_{(n+1)/2} + j \Delta)  + 
              \sum_{j=0}^{(n-3)/2} (r_{(n+3)/2} + j \Delta) 
       \right] \\
& = \frac{1}{n} 
       \left[ \frac{n+1}{2} \cdot r_{(n+1)/2} + 
              \Delta \binom{(n+1)/2}{2} + 
              \frac{n-1}{2} \cdot r_{(n+3)/2} + 
              \Delta \binom{(n-1)/2}{2}
       \right] \\
& =    \frac{r_{(n+1)/2} + r_{(n+3)/2}}{2} + 
       \frac{r_{(n+1)/2} - r_{(n+3)/2}}{2n} +
       \frac{\Delta}{n} \cdot \frac{(n-1)^2}{4}  \\
& \leq \frac{1+\Delta}{2} + \frac{\Delta}{n} \cdot \frac{(n-1)^2}{4} \\
& =    \frac{1}{2} + \Delta \frac{n^2 + 1}{4n} 
~.
\end{align*}
We have two cases. 
If $r_1 \geq 1$, then there are $n-1$ gaps of size $\Delta$, as
well as one gap of size at most $\Delta/2$. Since the gaps cover the
entire interval, we have that $(n-1)\Delta + \frac{\Delta}{2} \geq 1$.
It follows that $n\Delta \geq \frac{2n}{2n-1}$.  Thus, we can
demonstrate the same bound, since
\[
\frac{\RR'(x)}{\opt_0(x)} 
=    \frac{n/\overline{r}}{2/\Delta}
\geq \frac{n \Delta}{1 + \frac{(n^2+1) \Delta}{2n}} 
=    \frac{1}{\frac{1}{n\Delta} + \frac{1}{2} + \frac{1}{2n^2}} 
\geq \frac{2n^2}{3n^2 -n + 1} 
>    \frac{2}{3} 
~.
\]
Finally, we consider the case where $r_1 < 1$.  For some $\epsilon
\in (0,\Delta/2]$, we can set $r_{(n+1)/2} = \frac{1}{2} + \epsilon$.
Since sensors $(n+1)/2$ and $(n+3)/2$ are of distance $\Delta$ from
one another, it follows that
\[
r_{\frac{n+3}{2}} - r_{\frac{n+1}{2}} 
= \paren{1/2 + \Delta - \epsilon} - \paren{1/2 + \epsilon}
= \Delta - 2 \epsilon 
~.
\]
Moreover, we will show that $\epsilon \leq \Delta /4$, and thus
$r_{(n+3)/2} - r_{(n+1)/2} \geq \Delta/2$.  To see this, note first
that it follows from the definition of a stretch sequence and the
assumption that $r_1 < 1$ that $r_1 = r_{(n+1)/2} + \Delta (n-1)/2$
and $r_2 = r_{(n+3)/2} - \Delta (n-3)/2$.  Hence their difference is
\[\textstyle
r_1 - r_n 
= ( r_{\frac{n+1}{2}} + \half \Delta (n-1) ) - 
  ( r_{\frac{n+3}{2}} + \half \Delta (n-3) )
= r_{\frac{n+1}{2}} - r_{\frac{n+3}{2}} + \Delta
= 2\epsilon 
~.
\]
However since $1-\Delta/2 \leq r_n \leq r_1 < 1$, it must be the case
that $r_1-r_n \leq \Delta/2$, and this implies that $\epsilon \leq
\Delta/4$.

Finally, a computation similar to the one above reveals that
\[
\overline{r} 
\leq \frac{r_{\frac{n+1}{2}} + r_{\frac{n+3}{2}}}{2} + 
       \frac{r_{\frac{n+1}{2}} - r_{\frac{n+3}{2}}}{2n} +
       \frac{\Delta}{n} \frac{(n-1)^2}{4} 
\leq \frac{1+\Delta}{2} - \frac{\Delta}{4n} + 
       \frac{\Delta}{n}  \frac{(n-1)^2}{4} 
=    \frac{1}{2} + \frac{n \Delta}{4} 
~.
\]
As this is the same bound that we obtained in the even case, we
similarly achieve the same $2/3$ bound. 
\end{proof}

\subsection{Putting It All Together}

It remains only to connect the pieces we have accumulated in the previous three sections. 

\begin{lemma}
\label{lemma:strip}
$\RR'(x^j,b^j) \geq \frac{2}{3} T_j$, for every strip $j$.  
\end{lemma}

\begin{proof}
The result follows immediately from Lemmas~\ref{lemma:unit-sol},
\ref{lemma:unit-convex}, and ~\ref{lemma:ratio}.
\end{proof}

Our main result now follows from our construction. 

\begin{theorem}
\label{thm:main}
\rr is a $\frac{3}{2}$-approximation algorithm for \sosc.
\end{theorem}
\begin{proof}
First, observe that 
\[
\sum_j \RR(x^j,b^j) 
=    \sum_j \sum_{x_i \in x^j} \frac{b^j_i}{r_i}
=    \sum_i \inv{r_i} \sum_{j : x_i \in x^j} b^j_i
\leq \sum_i \inv{r_i} b_i 
=    \RR(x,b)
~.
\]
By Lemmas~\ref{lemma:convex} and~\ref{lemma:strip} we have that
\[
\RR(x,b) 
\geq \sum_j \RR(x^j,b^j)
\geq \sum_j \RR'(x^j,b^j)
\geq \sum_j \frac{2}{3} T_j
=    \frac{2}{3} T
=    \frac{2}{3} \opt(x,b)
~,
\]
and we are done.
\end{proof}


\subsection{Strip Cover}

Theorem~\ref{thm:main} readily extends to the \strip problem. 

\begin{theorem}
\rr is a $\frac{3}{2}$-approximation algorithm for \strip.
\end{theorem}

\section{Duty Cycle Algorithms}
\label{sec:dc}

In this paper we analyzed the \rr algorithm in which each
sensor works alone.  One may consider a more general version of this
approach, where a schedule induces a partition of the sensors into
sets, or \emph{shifts}, and each shift works by itself.  In \rr each
shift consists of one active sensor.  We refer to such an algorithm as a
\emph{duty cycle} algorithm.

In this section we show that, in the worst case, no duty cycle
algorithm outperforms \rr.  More specifically, we show that the
approximation ratio of any duty cycle algorithm is at least
$\threehalves$.

\begin{lemma}
\label{lemma:DC}
The approximation ratio of any duty cycle algorithm is at least
$\threehalves$ for both \sosc and \strip.
\end{lemma}
\begin{proof}
Consider an instance where $x = (\frac{1}{4}, \frac{3}{4},
\frac{3}{4})$ and $b = (2,1,1)$.  An optimal solution is obtained by
assigning $\rho_1 = \rho_2 = \rho_3 = \frac{1}{4}$, $\tau_1 = \tau_2 =
0$ and $\tau_3 = 4$.  That is, sensor $1$ covers the interval
$[0,0.5]$ for $8$ time units, sensors $2$ covers $[0.5,1]$ until time
$4$, and sensors $3$ covers $[0.5,1]$ from time $4$ to $8$.  This
solution is optimal in that it achieves the maximum possible lifetime
of $8 = 2 \sum_i b_i$.

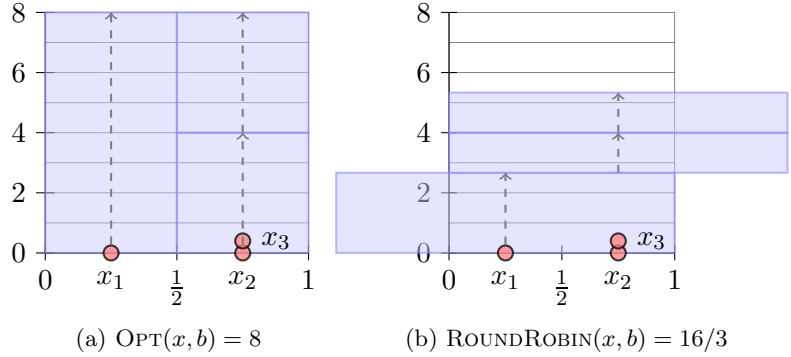
\begin{figure}[t]
\centering
\subfloat[$\opt(x,b) = 8$]{			\begin{tikzpicture}[xscale=3.5, yscale=0.4,
				coverage/.style={circle,draw=blue!50,fill=blue!20,thick,opacity=0.5},
				lifetime/.style={->,dashed,thick,gray},
				sensor/.style={circle,draw=black,fill=red!50,thick,opacity=0.75,inner sep=0pt,minimum size=2mm}]
				\def\T{8}
				\draw[step=1cm,gray,very thin] (0,0) grid (1, \T);	
				\draw[-] (0,0) -- (1,0) coordinate (x axis);
				\foreach \x/\xtext in {0, 0.5/\frac{1}{2}, 1}
					\draw (\x, 7pt) -- (\x, -7pt) node[anchor=north] {$\xtext$};
				\draw[-] (0,0) -- (0,8) coordinate (y axis);
				\foreach \y/\ytext in {0, 2,..., \T}
					\draw (1pt,\y) -- (-1pt,\y) node[anchor=east] {$\ytext$};
				\coordinate (a) at (0,0);	
				\coordinate (b) at (1,0);	
				\draw [coverage] (0,0) rectangle (1/2, \T);
				\draw [coverage] (1/2,0) rectangle (1, 4);
				\draw [coverage] (1/2,4) rectangle (1, \T);
				\node [sensor] (s1) at (1/4, 0) [label=below:$x_1$] {};
				\node [sensor] (s2) at (3/4, 0) [label=below:$x_2$] {};
				\node [sensor] (s3) at (3/4, 0.4) [label=right:$x_3$] {};
				\draw [lifetime] (s1) -- (1/4,\T);
				\draw [lifetime] (s2) -- (3/4,4);
				\draw [lifetime] (3/4,4) -- (3/4,\T);
			\end{tikzpicture}}
\subfloat[$\rr(x,b) = 16/3$]{			\begin{tikzpicture}[xscale=3, yscale=0.4,
				coverage/.style={circle,draw=blue!50,fill=blue!20,thick,opacity=0.5},
				lifetime/.style={->,dashed,thick,gray},
				sensor/.style={circle,draw=black,fill=red!50,thick,opacity=0.75,inner sep=0pt,minimum size=2mm}]
				\def\T{8}
				\draw[step=1cm,gray,very thin] (0,0) grid (1, \T);	
				\draw[-] (0,0) -- (1,0) coordinate (x axis);
				\foreach \x/\xtext in {0, 0.5/\frac{1}{2}, 1}
					\draw (\x, 7pt) -- (\x, -7pt) node[anchor=north] {$\xtext$};
				\draw[-] (0,0) -- (0,8) coordinate (y axis);
				\foreach \y/\ytext in {0, 2,..., \T}
					\draw (1pt,\y) -- (-1pt,\y) node[anchor=east] {$\ytext$};
				\draw [coverage] (-1/2,0) rectangle (1, 8/3);
				\draw [coverage] (0,8/3) rectangle (3/2, 12/3);
				\draw [coverage] (0,12/3) rectangle (3/2, 16/3);
				\coordinate (a) at (0,0);	
				\coordinate (b) at (1,0);	
				\node [sensor] (s1) at (1/4, 0) [label=below:$x_1$] {};
				\node [sensor] (s2) at (3/4, 0) [label=below:$x_2$] {};
				\node [sensor] (s3) at (3/4, 0.4) [label=right:$x_3$] {};
				\draw [lifetime] (s1) -- (1/4,8/3);
				\draw [lifetime] (3/4,8/3) -- (3/4,4);
				\draw [lifetime] (3/4,4) -- (3/4,16/3);
			\end{tikzpicture}}
\caption{Best schedule vs.\ best duty cycle schedule. 
Here $x = (\frac{1}{4}, \frac{3}{4}, \frac{3}{4})$ and $b = (2,1,1)$.}
\label{fig:DC}
\end{figure}

On the other hand, the best duty cycle algorithm is \rr, which
achieves a lifetime of $16/3$ time units.  (The shifts $\{1,2\}$ and
$\{3\}$ would also result in a lifetime of $16/3$ time units.)
Both schedules are shown in Figure~\ref{fig:DC}.
\end{proof}


\section{Set Radius Strip Cover}
\label{sec:srsc}

In this section we present an optimal $O(n^2 \log n)$-time algorithm
for the \srsc problem.
%
%
%
%
Recall that in \srsc we may only set the radii of the sensors since
all the activation times must be set to 0.  More specifically, we
assign non-zero radii to a subset of the sensors which we call \emph{active}, while the rest of
the sensors get $\rho_i=0$ and do not participate in the cover. 

Given an instance $(x,b)$, a radial assignment $\rho$ is called
\emph{proper} if the following conditions hold:
\begin{enumerate*}
\item Every sensor is either inactive, or exhausts its battery by 
      time $T$, where $T$ is the lifetime of $\rho$. That is, $\rho_i
      \in \{0,b_i/T\}$,
\item No sensor's coverage is superfluous. That is, for every active 
      sensor $i$ there is a point $u_i \in [0,1]$ such that $u_i \in
      [x_i-\rho_i,x_i+\rho_i]$ and $u_i \not\in
      [x_k-\rho_k,x_k+\rho_k]$, for every active $k \neq i$.
\end{enumerate*}

\begin{lemma}
\label{lemma:proper}
There is
a proper optimal assignment for every \srsc instance.
\end{lemma}
\begin{proof}
Let $I=(x,b)$ be a \srsc instance, and let $\rho$ be an optimal
assignment for $I$ with lifetime $T$.  We first define the assignment
$\rho' = b/T$ and show that it is feasible.  Since $\rho$ has lifetime
$T$, any point $u \in [0,1]$ is covered by some sensor $i$ throughout
the time interval $[0,T]$.  It follows that $\rho_i \leq b_i/T =
\rho'_i$.  Hence, $u \in [x_i-\rho'_i,x_i+\rho'_i]$, and thus $\rho'$
has lifetime $T$.
Next, we construct an assignment $\rho''$.  Initially, $\rho''=\rho'$.
Then starting with $i=1$, we set $\rho''_i=0$ as long as $\rho''$
remains feasible.  Clearly, $\rho''_i \in \{0,b_i/T\}$.  Furthermore,
for every sensor $i$ there must be a point $u_i \in
[x_i-\rho''_i,x_i+\rho''_i]$ such that $u_i \not\in
[x_k-\rho''_k,x_k+\rho''_k]$, for every active $k \neq i$, since
otherwise $i$ would have been deactivated.  Hence, $\rho''$ is a
proper assignment with lifetime $T$, and is thus optimal.
\end{proof}

Given a proper optimal solution, we add two dummy sensors, denoted $0$
and $n+1$, with zero radii and zero batteries at $0$ and at $1$,
respectively.  The dummy sensors are considered active.
%
%
We show that the optimal lifetime of a given instance is determined by
at most two active sensors.

\begin{lemma}
\label{lemma:pair}
Let $T$ be the optimal lifetime of a given \srsc instance $I=(x,b)$.
There exist two sensors $i, k \in \{0,\ldots,n+1\}$, where $i < k$,
such that $T = \frac{b_k+b_i}{x_k - x_i}$.
\end{lemma}
\begin{proof}
Let $\rho$ be the proper optimal assignment, whose existence is
guaranteed by Lemma~\ref{lemma:proper}.  We claim that there exist two
neighboring active sensors $i$ and $k$, where $i<k$, such that $\rho_i
+ \rho_k = x_k - x_i$.  The lemma follows, since $\rho_i = b_i/T$ and
$\rho_k = b_k/T$.

Observe that if $\rho_i + \rho_k < x_k - x_i$, for two neighboring
active sensors $i$ and $k$, then there is a point in the interval
$(x_i,x_k)$ that is covered by neither $i$ and $k$, but is covered by
another sensor.  This means that either $i$ or $k$ is redundant, in
contradiction to $\rho$ being proper. Hence, $\rho_i + \rho_k \geq x_k
- x_i$, for every two neighboring active sensors $i$ and $k$.

Let $\alpha = \min \set{\frac{\rho_k + \rho_i}{x_k - x_i} : i, k
\text{ are active}}$.  If $\alpha = 1$, then we are done.  Otherwise,
we define the assignment $\rho' = \rho/\alpha$.  $\rho'$ is feasible
since $\rho'_i + \rho'_k = \frac{1}{\alpha} (\rho_i + \rho_k) \geq x_k
- x_i$, for every two neighboring active sensors $i$ and
$k$. Furthermore, the lifetime of $\rho'$ is $\alpha T$, in
contradiction to the optimality of $\rho$.
\end{proof}

Lemma~\ref{lemma:pair} implies that there are $O(n^2)$ possible
lifetimes.  This leads to an algorithm for solving \srsc.

\begin{theorem}
\label{thm:srsc}
There exists an $O(n^2 \log n)$-time algorithm for solving \srsc.
\end{theorem}
\begin{proof}
First if $n = 1$, then $\rho_1 \gets r_1 \eqdf \max(x_1,1-x_1)$ and we
are done.
Otherwise, let $T_{ik} \gets \frac{b_k + b_i}{x_k-x_i}$, for every $i, k
\in \{0,\ldots,n+1\}$ such that $i<k$.  After sorting the set
$\set{T_{ik} : i<k}$, perform a binary search to find the largest
potentially feasible lifetime.  The feasibility of candidate $T_{ik}$
can be checked using the assignment $\rho^{ik}_\ell \gets
b_\ell/T_{ik}$, for every sensor $\ell$.

There are $O(n^2)$ candidates, each takes $O(1)$ to compute, and
sorting takes $O(n^2 \log n)$ time.  Checking the feasibility of a
candidate takes $O(n)$ time, and thus the binary search takes $O(n
\log n)$.  Hence, the overall running time is $O(n^2 \log n)$.
\end{proof}


\section{Discussion and Open Problems}

We have shown that \rr, which is perhaps the simplest possible
algorithm, has a tight approximation ratio of $\threehalves$ for both
\sosc and \strip.  We have also shown that \sosc is NP-hard, but it
remains to be seen whether the same is true for \strip.  Future work
may include finding algorithms with better approximation ratios for
either problem. However, we have eliminated duty cycle algorithms as
candidates.
Observe that both \sosc and \stsc are NP-hard, while \srsc can be
solved in polynomial time.  This suggests that hardness comes from
setting the activation times.

We have assumed that the battery charges dissipate in direct
inverse proportion to the assigned sensing radius (e.g. $\tau =
b/\rho$). It is natural to suppose that an exponent could factor into
this relationship, so that, say, the radius drains in quadratic
inverse proportion to the sensing radius (e.g. $\tau = b/\rho^2$).
One could expand the scope of the problem to higher
dimensions. Before moving both the sensor locations and the
region being covered to the plane, one might consider moving one but
not the other. This yields two different problems: 1) covering the
line with sensors located in the plane; and 2) covering a region of
the plane with sensors located on a line.





\bibliographystyle{abbrv}
\bibliography{references}


\end{document}